\documentclass[journal,12pt,draftclsnofoot,onecolumn]{IEEEtran}
\usepackage{cite}
\usepackage{enumerate}
\usepackage{color}
\usepackage{graphics}
\usepackage{epsfig,psfrag}
\usepackage{latexsym}
\usepackage{listings}
\usepackage{amssymb}
\usepackage{amsmath}
\usepackage{amsfonts}
\usepackage{epstopdf}
\usepackage{overpic}
\usepackage[normalem]{ulem}

\newtheorem{theorem}{Theorem}

\newtheorem{proposition}[theorem]{Proposition}

\newtheorem{definition}{Definition}

\newcommand{\ds}{\displaystyle}
\newcommand{\ra}{\rightarrow}

\newcommand{\blue}[1]{{\color{blue} #1}}

\renewcommand{\v}[1]{\ensuremath{\boldsymbol{#1}}}

\newcommand{\beqa}{\begin{eqnarray*}}
\newcommand{\eeqa}{\end{eqnarray*}}
\newcommand{\beqan}{\begin{eqnarray}}
\newcommand{\eeqan}{\end{eqnarray}}
\newcommand{\beqn}{\begin{equation}}
\newcommand{\eeqn}{\end{equation}}

\newcommand{\bE}{\boldsymbol{E}}

\newcommand{\bsy}{\boldsymbol}

\newcommand{\reals}{\ensuremath{\mathbf{R}}}
\mathchardef\mhyphen="2D

\begin{document}

\title{Heavy Traffic Approximation of Equilibria in Resource Sharing Games}

\author
{
Yu~Wu, Loc~Bui, and Ramesh~Johari
\thanks{Yu Wu and Ramesh Johari are with the Department of Management Science and Engineering, Stanford University, Stanford, CA 94305, USA (emails: \{yuwu, ramesh.johari\}@stanford.edu). Loc Bui is with the School of Engineering, Tan Tao University, Long An, Vietnam (email: locbui@ieee.org). Most of this work was done while Loc Bui was a postdoc at Stanford University.}
\thanks{A previous version of this paper was presented at the 7th Workshop on Internet \& Network Economics (WINE'11).}
\thanks{This work was supported by the National Science Foundation under
grants CMMI-0948434, CNS-0904609, CCF-0832820, and CNS-0644114, and by
the Defense Advanced Research Projects Agency under the ITMANET program.}
}


\maketitle

\begin{abstract}
We consider a model of priced resource sharing that combines
both {\em queueing behavior} and {\em strategic behavior}.  We study
a priority service model where a single  server allocates its capacity to
agents in proportion to their payment to the system, and users from different
classes act to minimize the sum of their cost for processing delay and payment.
As the exact processing time of this system is hard to compute and
cannot be characterized in closed form, 
we introduce the notion of {\em heavy traffic equilibrium}
as an approximation of the Nash equilibrium, derived by considering
the asymptotic regime where the system load approaches capacity.  We
discuss efficiency and revenue, and in particular provide a 
bound for the price of anarchy of the heavy traffic equilibrium.

%

\

\end{abstract}


\section{Introduction}\label{sec:intro}

A range of resource sharing systems, such as computing or
communication services, exhibit two distinct characteristics: {\em
  queueing behavior} and {\em strategic behavior}.  Queueing behavior
arises because jobs or flows are served with the limited capacity of
system resources.  Strategic behavior arises because these jobs or
flows are typically generated by self-interested, payoff-maximizing
users.  Analysis of strategic behavior in queueing systems has a long
history, dating to the seminal work of Naor \cite{nao69}; see the book
by Hassin and Haviv \cite{hashav03} for a comprehensive survey.
The interaction of queueing and strategic behaviors has
become especially important recently, with the rise of paid resource
sharing systems such as cloud computing platforms.
For example, \cite{allgur10} and \cite{chen10} discussed systems with
multiple service providers, modeled as first-come-first-serve queues,
that compete in both price and response time for potential buyers.


In this paper we consider a particular queueing model where a single server is shared among multiple jobs, 
and the service capacity allocated to each job depends on its priority level.
The particular scheduling policy we consider is known in the
literature as the {\em discriminatory processor sharing} (DPS) policy,
introduced by Kleinrock \cite{kle67}.
In the DPS model, the server shares its capacity 
{\em in proportion} to the priority level of  all jobs currently in the
system.  This service allocation rule is a special case of a more
general scheduling policy for queueing networks known as 
{\em proportionally fair} resource sharing \cite{kelmautan98,masrob00};
such scheduling policies have been studied extensively in the context of networked
resource sharing (see \cite{kankelleewil09, verayenun10} and references
therein). A survey of the DPS literature can
also be found in \cite{altavraye06}.


We consider a DPS system in steady state, and study a {\em job level}
game where every individual job is a single strategic user.  
Each user chooses a payment $\beta$,
which corresponds to the priority level of that user.  The user also
incurs a cost proportional to total processing time. 
The users' goal is to choose priority levels to minimize the sum of
expected processing cost and payment.  
(We also briefly discuss a {\em class level} game, 
where  every {\em class} is a single user.)

This game is inspired by resource sharing in real services. For
  example, in the Amazon EC2 Spot Instances, a user can bid her own
  price (priority) and enjoy the service as long as the dynamic
  benchmark price computed by the system is lower than the bid.  The
  service is terminated either upon completion of task or when the
  system price rises above the bid price.  In many file hosting
  websites, users can purchase premium packages which increase
  upload/download bandwidth and speed, and allow parallel tasks among
  other benefits.  Note that in these services, the resource is shared
  among all users that currently request service, and a higher payment
  leads to higher performance.

A central difficulty in analysis of equilibria arises 
because exact computation of the steady state processing time of a single job, 
given the priority choices of other jobs, is not possible in closed form.
Since the queueing behavior computation itself involves
numerical complexity, equilibrium characterization in closed form 
for the strategic behavior is essentially impossible. 
Thus obtaining structural insight into the games is
a significant challenge.

To tackle this problem,
we propose an approximate approach to
equilibrium characterization that is amenable to analysis, computable
in closed form, and provably exact in an appropriate asymptotic
regime where the load on the system increases, known as the {\em heavy
  traffic} regime \cite{kin62,rei84}.  The heavy traffic asymptotic regime is
widely used in analysis of queueing systems 
and is especially
valuable to study systems with many users. 
Asymptotics yield two benefits.  {\em First}, they significantly simplify
stochastic analysis.  The {\em second} key benefit of asymptotics is that
we are also able to simplify our game theoretic analysis.  Informally,
an important reason is that when the number of users grows large, no single 
user has a large impact on the whole system; this effect allows us to
simplify calculation of equilibria. (Note that this is similar to the
``large market'' approximation used to justify competitive price
equilibria in economics.)

We conclude with a brief survey of related work.
Priority pricing problems in queueing
systems with disciplines other than DPS have been investigated,
such as design of efficient, incentive compatible pricing for
nonpreemptive priority FIFO queues \cite{mend90,kim03}.
Besides the service discipline, our work also differs from this work
because in our model users choose their own priority levels, while in
previous models the service provider fixes the priority levels
available.
Other studies focus on
optimal arrival strategy of users, including the work by Glazer and
Hassin \cite{gla83} and Jain et al. \cite{jai11}.  In our paper,
although we briefly discuss the case where arrivals are endogenously
generated, in the main discussion we assume that arrival rates exogenously given.  

Turning our attention to DPS specifically, we note that most prior
work on DPS considers only analysis of the queueing system 
without any strategic choice of priorities. 
Haviv and van der Wal \cite{havwal97} consider a DPS system
in which users choose their own priority levels to minimize their costs, 
but only study a model with one class of users.
To the best of our knowledge, there is no previous work that considers
priority pricing in the multiple class DPS system, because expected
waiting time cannot be characterized in closed form.

Our main contributions are as follows.

%

{\em (1) An approximate notion of equilibrium.}  
Using an approximation to the processing time derived via the heavy
traffic asymptotic regime, we suggest a natural corresponding notion
of equilibrium that we call {\em heavy traffic equilibrium} (HTE).  
In an HTE, users minimize the sum of their payment and heavy traffic
processing time cost, rather than their true expected processing time
cost.
We show that under mild conditions, HTE exists and is unique, 
and that it can be computed in closed form in terms of
system parameters.  It is thus both simple to compute, and
asymptotically accurate when the system approaches heavy traffic.

{\em (2) Economic analysis: parameter sensitivity, efficiency, and revenue.} 
A significant benefit of our approach is that since we can compute the equilibrium
in closed form, it is straightforward to carry out analysis of efficiency and revenue.  
We study how the system behavior changes when cost or arrival rate parameters are scaled, 
and more importantly, we investigate social efficiency and system revenue of HTE 
under different system parameters, and give a 
bound for the price of anarchy of HTE.  We obtain some intriguing
insights: in particular, we show that within a particular class of
pricing schemes, and for a wide range of parameter choices, the
incentives of the revenue maximizing service provider become aligned
with minimization of total system processing cost.

The remainder of the paper is organized as follows. In Section~\ref{sec:model},
we describe the queueing game setup. 
In Section~\ref{sec:joblevel}, we introduce the notion of heavy traffic equilibrium. 
We then present the results on parameter sensitivity, efficiency, and revenue under the heavy traffic equilibrium in Section~\ref{sec:insights}. 
Some extensions of the model are discussed in Section~\ref{sec:dis},
followed by the conclusion. The proofs of the theorems are in the Appendix.

\section{Resource Sharing Game}\label{sec:model}


We consider a queueing game in which $K$ classes of jobs share
a single server of unit capacity.
Class $i$ ($i = 1, \cdots, K$) jobs arrive according to
a Poisson process with arrival rate $\lambda_i$
and have i.i.d.~exponentially distributed service requirements
(measured in units of service, e.g., processing cycles) with mean
$1/\mu_i$. 
We assume for simplicity that
$\mu_i=\mu$ for all classes.
 Let $\lambda=\sum_k\lambda_k$ denote the total arrival rate to the
system. Also, let $\rho_i=\lambda_i/\mu$ be the {\em load} of class
$i$, and define the {\em system load} as $\rho=\sum_k\rho_k = \sum_k
\lambda_k / \mu$. To ensure stability, we assume $\rho<1$. It is well
known that under this condition, the resulting queueing system
is ergodic and possesses a unique steady state distribution \cite{kle75}.
Waiting and being served in the system induces a cost $c_i$
per unit time for users of class $i$.
Without loss of generality we assume $c_1 > c_2 > \cdots > c_K$:
if two classes $i$ and $j$ have the same cost $c_i = c_j$, 
then they can be merged into one class with arrival rate $\lambda_i + \lambda_j$.

We assume that the server allocates its capacity according to the {\em
discriminatory processor sharing} (DPS) policy.  Under this policy,
each job is associated with a priority level.  If there are currently
$N$ jobs in the system and job $\ell$ has chosen priority level
$\beta_\ell$, then the fraction of service
capacity allocated to job $\ell$ is $\beta_\ell/\sum_{m=1}^N
\beta_m$.  


Upon arrival, without observing the state of the system, each job
chooses a priority level $\beta>0$. 
We consider a family of
pricing rules for priority that we refer to as {\em $\alpha$-fair
  pricing rules}, where $\alpha > 0$.  Formally, we assume that if a job
chooses priority level $\beta$, then the system manager charges that
job a price $\beta^\alpha$, where $\alpha > 0$.  Varying $\alpha$
allows us to study a range of pricing schemes.  In particular, as
$\alpha \to 0$, jobs face a strongly diminishing marginal cost to
higher choices of $\beta$; while as $\alpha \to \infty$, jobs face a
strongly increasing marginal cost with higher choices of $\beta$.

The pricing rules we consider are closely related to {\em $\alpha$-fair
allocation rules} studied in the networking literature
\cite{MoWalrand00}.  In an $\alpha$-fair allocation system,
one unit of resource is allocated to $N$ users, 
whose utility functions are characterized by $\alpha$: 
$U^{(\alpha)}(x)=x^{1-\alpha}/(1-\alpha)$ if
$\alpha\neq 1$, and $U^{(\alpha)}(x)=\log(x)$ if $\alpha=1$.
Users make payments for use of the system.  Let $w_\ell$ be the payment
of user $\ell$; the payments determine users' weights in the system.
Formally, suppose the payment vector of users is $\bsy{w}$ and  the allocation 
vector is $\bsy{x}$; then the resource manager solves the following optimization problem:
\[
\max_{\bsy{x}}~ \sum_{\ell = 1}^N w_\ell U^{(\alpha)}(x_\ell)
\quad
s.t.~ \sum_{\ell=1}^N x_\ell \leq 1.
\]
The solution of this problem is $x_\ell=w_\ell^{1/\alpha}/(\sum
w_m^{1/\alpha})$.  A well-known example of an $\alpha$-fair
allocation rule is the {\em proportionally fair} allocation rule,
obtained when $\alpha = 1$ \cite{kelmautan98}: resource is allocated proportional to payment. 
Now, suppose that the $\alpha$-fair pricing rule is used in our model, so that $w_\ell =
\beta_\ell^\alpha$.  Then the $\alpha$-fair allocation rule reduces to the
discriminatory processor sharing policy described above---i.e.,
allocation of server capacity in proportion to the priority levels $\beta_\ell$.


In this paper we will generally be interested in scenarios  where all
jobs of the same class $i$ choose the same priority level.  In an
abuse of notation we denote by $\beta_i$ the priority level chosen by
all class $i$ jobs, and in this case we succinctly denote $(\beta_1, \cdots, \beta_K)$ by
$\bsy{\beta}$.  We refer to $\bsy{\beta}$ as the {\em class priority vector}.
Let $V(\beta; \bsy{\beta})$ be the expected processing time for a job
with priority $\beta$ that arrives to the system in steady state,
with the class priority vector given by $\bsy{\beta}$.  Observe that
with this notation, a class $i$ job with priority level $\beta_i$ has
expected processing time $V(\beta_i; \boldsymbol{\beta})$.  For
convenience we define $W_i(\bsy{\beta}) = V(\beta_i; \bsy{\beta})$.
The total cost of a user is $cV(\beta; \bsy{\beta})+\beta^{\alpha}$,
where $c$ is the user's unit time cost and $\beta$ is its priority level.

We frequently make use of {\em Little's law}, which provides a
relationship between steady state expected processing times and steady
state expected queue lengths \cite{kle75}.  In particular, let $N_i$ denote the steady state
number of class $i$ jobs in the system.  In a system consisting of $K$
classes $(\lambda_i, \beta_i)$, $i = 1,\ldots,K$,
Little's law establishes that in steady state, for every class $i$ we have
$\bE[N_i] = \lambda_i W_i(\bsy{\beta})$.

\subsection{Nash Equilibrium}

We consider two types of games for this system: the {\em job level game} and the
{\em class level game}. 
In the job level game, each {\em job} is an individual
user, aiming to minimize its expected total cost by choosing its
own priority level $\beta$. 
Although jobs from the same class are allowed to choose different
priority levels, because jobs of the same class share the same
parameters {\em ex ante}, we restrict our attention only to {\em symmetric}
equilibria of the resource sharing game; these are equilibria where jobs from
the same class choose the same priority levels.  
Such an equilibrium can be characterized by a class priority vector
$(\beta_1, \cdots, \beta_K)$.

\begin{definition}\label{def:jobeq}
A {\bf job level Nash equilibrium}
consists of a class priority vector $\boldsymbol{\beta}=(\beta_1,
\cdots, \beta_K)$ such that for all $i=1, \cdots, K$,
\begin{equation}\label{eq:def:jobeq}
\beta_i =\arg\min_{\beta>0}\left[c_iV(\beta;
\boldsymbol{\beta})+\beta^{\alpha}\right],\ \forall\ i=1, \cdots, K.
\end{equation}
\end{definition}

In the class level game, each class is regarded as a single user 
and chooses a priority level for all of its jobs;
therefore the equilibrium is again characterized by a class priority
vector.

\begin{definition}\label{def:classeq}
A {\bf class level Nash equilibrium}
consists of a class priority vector $\boldsymbol{\beta}=(\beta_1,
\cdots, \beta_K)$ such that for all $i=1, \cdots, K$,
\beqan
\beta_i = \arg\min_{\beta\geq\underline{\beta}}\ [c_i W_i(\beta_1, \cdots,
\beta_{i-1}, \beta, \beta_{i+1},
\cdots, \beta_K)+\beta^{\alpha}],\ \forall\ i=1, \cdots, K.
\label{eq:def:classeq}
\eeqan
\end{definition}



We emphasize that, although jobs from the same class choose
the same priority in both the symmetric job level equilibrium and
the class level equilibrium, these two equilibria are not identical.
The difference
is that in the class level game, changing the priority level of a whole
class $i$ causes an externality within the class itself, 
while by contrast, in the job level game, a single
job alters its priority level in isolation.
In this paper, we mainly study the job level game, 
but also briefly discuss how our study can be adapted to the class level game
in the Section \ref{sec:dis}.

\subsection{Characterizing Processing Times}

Nash equilibria of both the job level and class level require the
characterization of processing times $V$ and $W_i$,
which is in general quite complex.
For the $K$ class DPS model, Fayolle et al.~\cite{faymitias80} show that
the expected steady state processing time $W_i$ for each class $i$ can
be determined by solving a linear system.

\begin{theorem}\label{thm:wtime} \cite{faymitias80}
In a $K$-class DPS model with class priority vector $\bsy{\beta}$,
$(W_1(\boldsymbol{\beta}), \cdots$ $, W_K(\boldsymbol{\beta}))$ is the unique
solution of the following system of equations:
\begin{equation}\label{eq:wtime}
\mu W_k(\boldsymbol{\beta})-
\sum_{i=1}^K\frac{\lambda_i\beta_i}{\beta_i+\beta_k}[W_k(\boldsymbol{\beta})+W_i(\boldsymbol{\beta})]=1,\
k=1, \cdots, K.
\end{equation}
\end{theorem}


On the other hand, computing the processing time $V(\beta; \bsy{\beta})$ for general $\beta$
can be reduced to computing the processing time
$W_i(\bsy{\beta}), i=1,\cdots, K$ as stated by the following theorem.

\begin{theorem}\label{thm:vmulti}
Let $N_i$ be the steady state number of class $i$ jobs in a $K$-class DPS system
with class priority vector $\bsy{\beta}$.  Then the steady state
processing time of a job with priority $\beta$ is
\beqan\label{eq:vmulti} V(\beta; 
\boldsymbol{\beta})=U_0(\beta;
\boldsymbol{\beta})+\sum_{i=1}^KU_i(\beta;
\boldsymbol{\beta}) \bE[N_i], \eeqan
where
\begin{equation} \label{eq:uk}
U_i(\beta; \boldsymbol{\beta})=\frac{\beta_i}{\beta_i+\beta}U_0(\beta; \boldsymbol{\beta}),\ i=1, \cdots, K;
\ \text{and}\ 
U_0(\beta; \boldsymbol{\beta})=\left[\mu-\sum_{i=1}^K\frac{\lambda_i\beta_i}{\beta_i+\beta}\right]^{-1}
\end{equation}
\end{theorem}
The values of $\bE[N_i]$ can be obtained by applying
Little's law to the solution of the system of linear equations (\ref{eq:wtime}). 
We conclude, therefore, that solving for $V(\beta; \boldsymbol{\beta})$
in (\ref{eq:vmulti}) can be reduced to computing $W_i(\bsy{\beta})$.
In general, explicitly solving (\ref{eq:wtime}) requires the inversion
of a $K\times K$ matrix with complexity $O(K^3)$, and hence, there is no closed form expression for $W_i$ or $V$.

Nevertheless, when $K=1$ or $K=2$, we are able to solve for $W_i$ and $V$ in closed form. 
The solution of (\ref{eq:vmulti}) with $K=1$ is first established in \cite{havwal97,hashav03} as follows
\begin{equation}\label{eq:vtwo2}
V(\beta;
\hat{\beta})=\frac{1}{\mu(1-\rho)}\cdot\frac{\beta(1-\rho)+\hat{\beta}}{\hat{\beta}(1-\rho)+\beta}.
\end{equation}
When $K=2$, the solution for $W_i$ is given by \cite{faymitias80},
and the solution for $V$ directly follows. 
Both solutions are lengthy and omitted for brevity.

\subsection{Existence of NE}

Existence of Nash equilibrium 
can be guaranteed when $\alpha\geq 1$, by exploiting
convexity of the job cost function in \eqref{eq:def:jobeq}. 
When $\alpha<1$, the payment term $\beta^{\alpha}$ is strictly concave, 
therefore the convexity of the objective function is not guaranteed.
Although analytically establishing existence of Nash equilibrium in this regime remains an open question,
our numerical computation with best response dynamics converges to a
  Nash equilibrium 
even when $\alpha<1$.

\begin{theorem}\label{thm:ne_existence}
There exists
a Nash equilibrium for the job level game when $\alpha \geq 1$.
\end{theorem}

As usual, this existence result is nonconstructive, since it uses a
fixed point theorem.  In general, given the implicit equations that
define the processing times in \eqref{eq:wtime}, there is no closed form characterization of
the Nash equilibrium, and no tractable approach for computation
is available.
Although we could resort to some heuristics (e.g., best response 
dynamics) to approach NE, each step of such an algorithm requires
computing a range of processing times with fixed parameters, and as
established above each such computation has complexity $O(K^3)$.
Further, there is no theoretical guarantee that such dynamics will converge.
(Though we note, numerical computation suggests that the best
  response dynamics does converge.)
Equilibrium computation is therefore not possible in closed form
in general; as a result, we are left with essentially no structural
insight into the behavior of players in the game.


\section{Heavy Traffic Approximation}\label{sec:joblevel}

In the remainder of the paper we consider an alternate approach to the
equilibrium analysis, by approximating the processing time.
We aim to overcome the complexity of
computing the processing times by exploiting a {\em heavy traffic}
approximation, i.e., an approximation where the load approaches
service capacity.  Such an approximation is relevant for large
systems such as cloud computing services, where providers will
typically not want to provision significant excesses of capacity.
\blue{\footnote{One such justification  for heavy traffic capacity provisioning 
comes from Nair et al. \cite{nai11},
who study optimal capacity provisioning for online service providers.
They find that as the market size becomes large,
heavy traffic emerges as a consequence of a profit maximizing strategy
for the service provider, with exact scaling depending on the strength
of positive externalities among users.}}


\subsection{Approximating the Processing Time}
\label{sec:heavy}

In heavy traffic, a phenomenon known as {\em state space collapse}
gives us a simplified solution for the steady state 
distribution of the system \cite{rei84}; informally, state space
collapse refers to the fact that the numbers of jobs of each class in the system become
perfectly correlated when the system is heavily loaded.

In a slight abuse, whenever
we write $\rho\rightarrow1$, we mean that we consider a sequence of
systems such that $(\rho_1, \cdots, \rho_K)$ 
converges to some $(\overline{\rho}_1, \cdots, \overline{\rho}_K)$ with 
$\sum_{i=1}^K\overline{\rho}_i = 1$. Moreover, we
emphasize that both $V(\beta; \bsy{\beta})$ and $W_i(\bsy{\beta})$
depend on $\bsy{\rho}$, though we suppress this dependence for notational brevity.
Let $N_i$ denote the steady state number of type $i$ jobs in the system. 
Then we have the following result on the joint
steady state distribution of $(N_1, \cdots, N_K)$ for a DPS system in
heavy traffic.
\begin{theorem}\label{thm:heavy}\cite{regsen96}
Let $N_i$ be the steady state number of class $i$ jobs in a $K$-class 
DPS system with class priority vector $\bsy{\beta}$. Then
as $\rho \to 1$,  we have
\begin{equation}\label{eq:heavy}
(1-\rho)(N_1, \cdots, N_K)\stackrel{d.}{\rightarrow}
Z\cdot\left(\frac{\overline{\rho}_1}{\beta_1}, \cdots,
\frac{\overline{\rho}_K}{\beta_K}\right),
\end{equation}
where ``$\stackrel{d.}{\rightarrow}$'' denotes convergence in distribution, 
and $Z$ is an exponentially distributed random variable with
mean $1/\overline{\gamma}(\bsy{\beta})$ where $\overline{\gamma}(\bsy{\beta})=\sum_{i=1}^K\overline{\rho}_i/\beta_i$.
\end{theorem}

Convergence of the joint distribution implies convergence of
marginal distributions, so
$(1-\rho)N_i\stackrel{d.}{\rightarrow} Z\overline{\rho}_i/\beta_i$
for each $i$. Moreover, the second moment of $(1-\rho)N_i$ is shown to be
uniformly bounded \cite{regsen96}, so the $N_i$'s are uniformly integrable.
It follows from \cite{bil87} that in this case convergence in distribution implies convergence in the mean, 
and hence,
\begin{equation}\label{eq:eni}
(1-\rho)\bE[N_i]\rightarrow
\bE[Z]\frac{\overline{\rho}_i}{\beta_i}=\frac{\overline{\rho}_i}{\beta_i\overline{\gamma}(\bsy{\beta})} \quad \mbox{as} \quad \rho \ra 1.
\end{equation}
Taking advantage of this approximation of $\bE[N_i]$, we are now able to
approximate $V(\beta;
\boldsymbol{\beta})$. 
Note that 
\[
\lim_{\rho\ra1}U_0(\beta; \bsy{\beta})=\left[\mu\left(1-\sum_{i=1}^K\frac{\overline{\rho}_i\beta_i}{\beta_i+\beta}\right)\right]^{-1}=\left(\mu\beta\sum_{i=1}^K\frac{\overline{\rho}_i}{\beta_i+\beta}\right)^{-1}
\]
is finite, so substituting (\ref{eq:uk}) and (\ref{eq:eni}) into
(\ref{eq:vmulti}) yields
\begin{align}
\lim_{\rho\rightarrow 1}(1-\rho) V(\beta;\boldsymbol{\beta})
&=\left(\lim_{\rho\ra 1}U_0(\beta;\bsy{\beta})\right)\left(\sum_{i=1}^K\frac{\beta_i}{\beta_i+\beta}\lim_{\rho\ra 1}(1-\rho)\bE[N_i]\right)
=\frac{1}{\mu\beta\overline{\gamma}(\bsy{\beta})}.
\label{eq:apporx:vbeta}
\end{align}
In the light of the above approximation, we have the following definition.
\begin{definition}
\label{def:ht_wait_time}
The {\bf heavy traffic processing time} for a job with priority level
$\beta$ in a system with $K$ classes with class priority vector
$\bsy{\beta}$ is defined as
\begin{equation}\label{eq:vht}
V^{HT}(\beta;\bsy{\beta})=\frac{1}{(1-\rho)}\cdot\frac{1}{\mu\beta
  \gamma(\bsy{\beta})}, \quad\mbox{where}\quad \gamma(\bsy{\beta}) = \frac{1}{\rho}\sum_{i=1}^K\frac{\rho_i}{\beta_i}.
\end{equation}
\end{definition}
%
%
We note that $V^{HT}(\beta;\bsy{\beta})$ has a closed form, and is easy 
to compute. Moreover, it is asymptotically exact in the heavy traffic regime: it is straightforward to show that as $\rho \ra 1$, $\gamma(\bsy{\beta}) \ra \overline{\gamma}(\bsy{\beta})$, and hence, 
$(1-\rho)[V^{HT}(\beta; \bsy{\beta})-V(\beta;\bsy{\beta})]\ra 0$.

We note here that one reason we consider the case where $\mu_i =
  \mu$ for all $i$ is that in the absence of this assumption, a similar result to
  Theorem \ref{thm:vmulti} becomes more challenging (in particular,
  because \eqref{eq:recursion} in the appendix is no longer
  tractable).  However, an appropriate generalization of Theorem \ref{thm:heavy}
  holds even for heterogeneous $\mu_i$, and based on this fact we
conjecture that the analysis of this paper can be carried out even
with heterogeneity of $\mu_i$.  We leave this for future work.

\subsection {Heavy Traffic Equilibrium}


For general $K$ it is quite hard to solve for pure Nash equilibrium
because: (i) computing $V(\beta; \boldsymbol{\beta})$ 
requires matrix inversion to solve the linear system
(\ref{eq:wtime}), which can only be done numerically; and (ii) even if
we are able to solve $V(\beta;\boldsymbol{\beta})$ numerically and
obtain optimality conditions for each player (which cannot be done in
closed form), we would still need to solve
a generally {\em nonlinear} system with $K$ equations and
$K$ unknowns to compute the Nash equilibrium.

In this section,  we propose a novel
concept of equilibrium which can be used to approximate the Nash
equilibrium, yet can be computed in closed form.  We approximate
$V(\beta;\bsy{\beta})$ by $V^{HT}(\beta;\bsy{\beta})$ in the objective
function, and based on this approximation we define a concept of equilibrium that
we call {\em heavy traffic equilibrium} (HTE) for job level games, as follows.

\begin{definition}\label{def:approxe}
A {\bf heavy traffic equilibrium} of the game
consists of a set of priorities
$\bsy{\beta}=(\beta_1,\cdots,
\beta_K)$ such that
\[
\beta_i=\arg\min_{\beta>0}\left(c_iV^{HT}(\beta;\bsy{\beta})+\beta^{\alpha}\right), ~ i=1, \cdots, K.
\]
\end{definition}

We can explicitly compute the heavy traffic equilibrium.
\begin{theorem}
\label{thm:ht_equi}
A heavy-traffic equilibrium always exists, and it is unique. 
Moreover, it can be calculated in closed form:
\begin{equation}\label{eq:opt_beta_tilde}
\beta_i=
c_i^{\frac{1}{\alpha+1}} [\alpha(1-\rho)\rho^{-1}S_1]^{-\frac{1}{\alpha}},
\end{equation}
where $S_1=\sum_{i=1}^K \lambda_i c_i^{-\frac{1}{\alpha+1}}$.
\end{theorem}

We have two remarks on this result.
First, this closed form expression allows us to carry out  analysis
of sensitivity, efficiency, and revenue of the HTE (see Section~\ref{sec:insights}).
Second, the HTE is easily computable with complexity $O(K)$.
In comparison, the complexity for computing the exact processing time
with fixed parameters is $O(K^3)$, 
and as discussed computing exact NE is intractable.

We have observed above that the difference between the heavy traffic processing time
and the exact processing time approaches zero as $\rho\to 1$, when
scaled by a factor $1-\rho$.  Using this approximation, we can also
prove an approximation theorem for the heavy traffic equilibrium:
we show that deviating by any constant factor from the HTE 
is not profitable as $\rho\ra 1$.

\begin{theorem}\label{thm:epsilon_approx}
Consider a sequence of systems indexed by $n$ such that classes have the same service capacity $\mu$,
and the loads of the systems $\rho^{(n)}\rightarrow 1$ as $n\rightarrow \infty$. 
Let $\boldsymbol{\beta}^{(n)}$ be the unique HTE of the $n$-th system, then for any $\delta\geq 0$,
\beqn
\lim_{n\rightarrow\infty} (1-\rho^{(n)})
\Big[ c_iV_{(n)}\left(\beta_i^{(n)}; \boldsymbol{\beta}^{(n)}\right)
+\left(\beta_i^{(n)}\right)^{\alpha}
 -c_iV_{(n)}\left(\delta\beta_i^{(n)}; \boldsymbol{\beta}^{(n)}\right)
-\left(\delta\beta_i^{(n)}\right)^{\alpha}\Big]
\leq 0.
\eeqn
Here $V$ is subscripted by $(n)$ to indicate that the processing time is computed in system $n$ with load $\rho^{(n)}$.
\end{theorem}


In the theorem, we consider deviations by a multiplicative constant factor
  rather than by an additive constant  
because (\ref{eq:opt_beta_tilde}) implies that, as $\rho \to 1$, the
heavy traffic equilibrium increases without bound; as a result, it is
straightforward to check that any additive constant deviation has no
beneficial effect as $\rho$ approaches $1$.
Note that the processing time is only asymptotically exact up to
a $1-\rho$ scaling, thus the same is true for this approximation theorem as
well.  Indeed, this is what we give up by studying heavy traffic:
while we gain analytical tractability, the ``resolution'' to which we
can study deviations is scaled by $1-\rho$.  This tradeoff is
systematic throughout the study of large scale queueing models
even without strategic behavior.

\subsection{Numerics: Approximation Error}

In this subsection, we numerically study the approximation error
  between the HTE and the exact NE, with different system parameters
  $(K, \{c_i\}, \{\lambda_i\}, \alpha, \rho)$; this complements our
  theoretical analysis above.  Given a HTE $\bsy{\beta}^{HT}$ and an NE
  $\bsy{\beta}^{NE}$, we use relative error as a measure of
  approximation, i.e.,
  $\max_i(\beta_i^{HT}-\beta_i^{NE})/\beta_i^{NE}$.  We compute NE
  using best response dynamics; surprisingly we found that best
  response dynamics converge to NE for all parameter choices below.



\begin{figure}
\begin{center}
\includegraphics[width=6in]{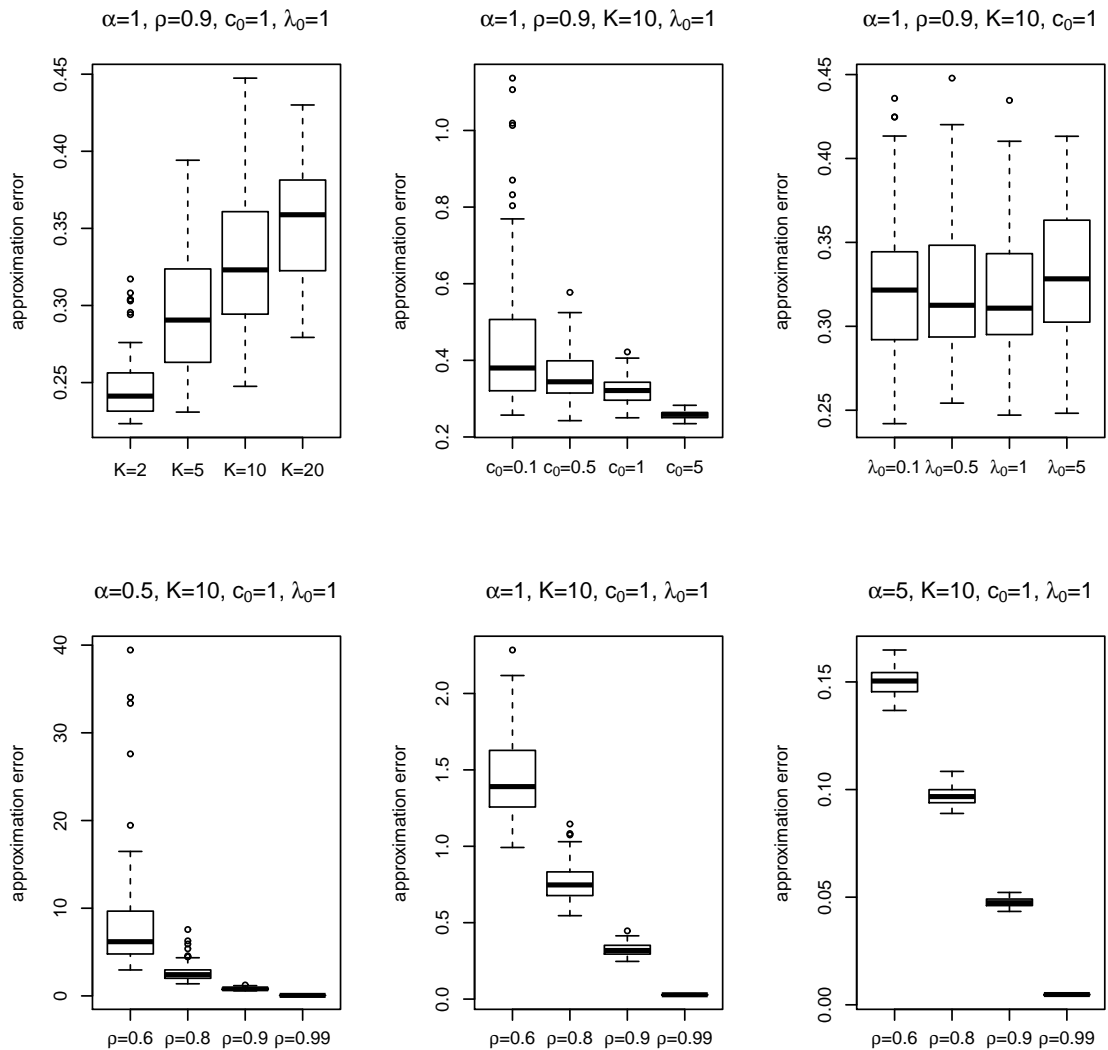}
\caption{Relative error of HTE vs. NE under different parameter choices.\vspace{-0.5in}} 
\label{fig:numerics}
\end{center}
\end{figure}

Figure \ref{fig:numerics} shows the impact of changes in system
parameters on approximation accuracy.  To quantify the heterogeneity
of $c_i$'s, we assume $c_i$'s are i.i.d.~ from a uniform distribution
on $[0,10]$, plus a constant $c_0$.
A smaller $c_0$, therefore, induces a potentially larger ratio between
the smallest and largest $c_i$. 
Similarly, we assume $\lambda_i$'s are i.i.d.~drawn from uniform$[0,
  10]+\lambda_0$.  To illustrate the change, we fix most of these
parameters at $K=10, c_0=\lambda_0=1, \alpha=1, \rho=0.9$, but vary
one or two of them at a time.  For each set of parameters, we have 100
simulation samples (of cost vector and arrival rate vector), and the
approximation errors are summarized by the boxplots.

In the upper panel of Figure \ref{fig:numerics}, 
we see that heterogeneity in the system weakens the approximation.
Numerical results show that approximation error is higher with smaller
$c_0$ or a larger number of classes; both increase heterogeneity.
Heterogeneity in the arrival rates appears to cause less degradation
in the approximation, but it can be shown that {\em both} arrival rate
heterogeneity and significant cost heterogeneity together can amplify
approximation errors.


In the lower panel of Figure  \ref{fig:numerics},
we vary $\alpha$ and $\rho$.
The results suggests that the approximation error is lower with larger $\alpha$.
Note that $\alpha$ describes the marginal cost to payment;
therefore larger $\alpha$ induces smaller payments and the relative heterogeneity among user decisions diminishes.
Regarding $\rho$, the error decreases as we approach heavy traffic
($\rho$ close to 1), as expected.

We conclude by noting that the approximation error can be made arbitrarily
large through appropriate parameter choices; one such example is given
when $K = 2$, $c_1/c_2 \ra\infty$ and $\lambda_1/\lambda_2\ra\infty$.
In this case since
$\lambda_2$ and $c_2$ are relatively small, the class $2$
optimal priority level in the exact NE is also extremely small.
However, in the heavy traffic approximation, the processing time
$V^{HT}(\beta, \bsy{\beta})$ is inversely proportional to $\beta$; 
so in HTE no user will chooses an extremely small $\beta$, leading
to arbitrarily large error rate.



\section{Sensitivity, Efficiency and Revenue}
\label{sec:insights}

The tractability of heavy traffic equilibrium allows us to
analytically study parameter sensitivity, as well as efficiency and
revenue at the HTE equilibrium.  
Throughout this section, we let
$\bsy{\beta}^*$ denote the HTE.

\subsection{Sensitivity} 

In this subsection, we analyze the sensitivity of the HTE, i.e., how
the equilibrium behaves with respect to changes in system
parameters. These observations follow directly from
(\ref{eq:opt_beta_tilde}).

{\em Sensitivity with respect to $c$.} If all $c_i$ are scaled by a constant $\zeta > 0$, 
then every $\beta^*_i$ is scaled by $\zeta^{\frac{1}{\alpha}}$. 
This is rather intuitive since the objective function is the sum of
expected processing cost and $\beta_i^\alpha$, and the expected
processing cost does not change any $V_i$.
Therefore the equilibrium is the same up to a scaling factor.
Further more, simple first derivative analysis shows that 
all equilibrium $\beta_i$'s are increasing in any single $c_j$.
Note that 
increasing $c_j$ provides extra incentive for class $j$ users to
invest in priority $\beta_j$.  
In equilibrium, all other $\beta_i$'s also increase as a result of
priority competition in the server. 


{\em Sensitivity with respect to $\rho$.} The ratio $\beta_i^*/\beta_j^*=(c_i/c_j)^{\frac{1}{\alpha+1}}$ is independent of $\rho$,
i.e., changing $\rho$ will change each $\beta^*_i$
but will not affect $\beta_i^*/\beta_j^*$ for any $i, j$.
Therefore the ratio between service capacity allocated to any pair of jobs,
as well as the ratio between the heavy traffic processing times of a pair of jobs, are invariant to the load of the system.

{\em Sensitivity with respect to $\alpha$.} When $\alpha\rightarrow 0$, every $\beta^*_i\rightarrow\infty$; when $\alpha\rightarrow\infty$, every $\beta^*_i\rightarrow1$. 
This is due to the fact that as $\alpha \ra 0$, jobs face a strongly
diminishing marginal cost to higher choices of $\beta$, and hence,
prefer to choose higher $\beta$ at the equilibrium; while the effect
is reversed as $\alpha \ra \infty$.


\subsection{Efficiency} 

In HTE, efficiency is characterized by the expected total cost incurred to the system in one unit of time:
\beqn\label{eq:efficiency}
\mathcal{C}=\sum_{i=1}^K\lambda_ic_iV^{HT}(\beta^*_i; \bsy{\beta}^*)
=\left(\frac{\rho}{1-\rho}\right)\left(\sum_{i=1}^K\lambda_ic_i^{\frac{\alpha}{\alpha+1}}\right)\left(\sum_{i=1}^K\lambda_ic_i^{-\frac{1}{\alpha+1}}\right)^{-1}.
\eeqn
We call $\mathcal{C}$ the {\em system processing cost} (a more efficient system has a lower value of $\mathcal{C}$).
Given fixed $\lambda_i$ and $c_i$ $(i = 1, \cdots, K)$, 
the efficiency depends on the system parameter $\alpha$ and the load $\rho$ as follows.

{\em Dependence of $\mathcal{C}$ on $\rho$}. We note that
$\mathcal{C}$ is proportional to $\rho/(1-\rho)$, and hence is
increasing in $\rho$.  
This is because a larger load $\rho$ implies a busier system, 
and therefore the processing time is longer.  (Note that we fixed
$\bsy{\lambda}$, so varying $\rho$ is equivalent to varying $\mu$.)

{\em Dependence of $\mathcal{C}$ on $\alpha$}. It is well known that
the system optimal scheduling policy is the 
$c\mhyphen\mu$ rule \cite{coxsmi61}: classes are given strict priority
in descending order of $c_i\mu_i$ (or equivalently in this paper, in
descending order of $c_i$, since we assume that all $\mu_i$ are
the same). That is, for any $1\leq i, j\leq K$, class $j$ jobs are
preempted by class $i$ jobs if $c_i\mu_i > c_j\mu_j$. Jobs with the
same value of $c\mu$ are served in first-in-first-out (FIFO) scheme. 
Since $\beta_i^*/\beta_j^*=(c_i/c_j)^{\frac{1}{\alpha+1}}$, for $c_i>c_j$,
the ratio $\beta_i^*/\beta_j^*$ is higher with smaller $\alpha$, so we
expect higher $\alpha$ lead to less efficient equilibria.
This intuition is analytically stated in the following theorem.

\begin{proposition}\label{thm:eff_alpha}
The HTE system processing cost $\mathcal{C}$ is increasing in $\alpha>0$.
\end{proposition}

We note that even when $\alpha$ approaches 0, the HTE does not
approach social optimum. In fact, for any $i,j$ such that $c_i>c_j$,
we have that $\beta_i^*/\beta_j^*=(c_i/c_j)^{\frac{1}{\alpha+1}}$, and
hence $1<\beta_i^*/\beta_j^*<c_i/c_j$. On the other hand, with the
$c\mhyphen\mu$ rule, if $c_i>c_j$, then class $i$ jobs completely
preempt class $j$ jobs, which can be interpreted as the case where
$\beta_i^*/\beta_j^* = \infty$. 
Therefore, it is clear that the HTE can never be as efficient as the
$c\mhyphen\mu$ rule, for any choice of $\alpha$.
However, we can upper bound the {\em price of anarchy} (PoA) of the HTE, as
stated in the following theorem.  The PoA is the ratio
$\mathcal{C}/\mathcal{C}^{opt}$, where $\mathcal{C}^{opt}$ is the
minimum expected system processing cost (achieved by the $c\mhyphen\mu$ rule).

\begin{theorem}\label{thm:hte_poa}
The price of anarchy (PoA) of the HTE is upper-bounded by:
\beqn
\frac{\mathcal{C}}{\mathcal{C}^{opt}} <
\frac{  \sum_{i=1}^{K-1} (\lambda_i/\lambda_K)(c_i/c_K)^{\frac{\alpha}{\alpha+1}} + 1 }
{  \sum_{i=1}^{K-1} (\lambda_i/\lambda_K)(c_i/c_K)^{-\frac{1}{\alpha+1}}  +1  }
< \left(\dfrac{\lambda-\lambda_K}{\lambda_K}\right)\left(\dfrac{c_1}{c_K}\right)^{\frac{\alpha}{\alpha+1}}+1.
\eeqn
\end{theorem}

Note that the upper bound can be made arbitrarily large through
appropriate parameter choices; further, this is tight, in the sense
that there exist systems where the PoA of HTE is in fact arbitrarily
large. For example, let $\lambda_i = \lambda$ for all $i$, and
set $c_K = 1$, $c_i=m$ for $i=1, \cdots, K-1$, and choose $\mu$ so that
$\rho=1-m^{-2}$. Then it can be shown that 
$\mathcal{C}/\mathcal{C}^{opt} =
\Omega\left((K-1)m^{\frac{\alpha}{\alpha+1}}\right)$ as $m \to \infty$
(see the proof of the theorem for details).  

We also note that the PoA bound is increasing in $\alpha$, 
which matches the intuition that a scheme closer to strict priority in
descending cost order yields higher social welfare.  
If we let $\alpha\rightarrow0$, then the PoA is asymptotically bounded by $\lambda/\lambda_K$.
In that case, if the arrival rates of all classes are the same, then
the PoA is bounded by $K$. 
We can also let $\lambda_1, \cdots, \lambda_{K-1}\rightarrow 0$ to
make the PoA approach $1$, but this is not surprising since in this
case the system essentially consists of only one class.

\subsection{Revenue}
The revenue of the server per unit time is the sum of expected payments in one unit of time:
\begin{equation}\label{eq:revenue}
\mathcal{R}=\sum_{i=1}^K\lambda_i(\beta_i^*)^{\alpha}
=\left(\frac{\rho}{\alpha(1-\rho)}\right)\left(\sum_{i=1}^K\lambda_ic_i^{\frac{\alpha}{\alpha+1}}\right)\left(\sum_{i=1}^K\lambda_ic_i^{-\frac{1}{\alpha+1}}\right)^{-1}.
\end{equation}
Given fixed $\lambda_i$ and $c_i$ $(i = 1, \cdots, K)$, 
the revenue depends on the system parameter $\alpha$ and the load $\rho$ as follows.

{\em Dependence of $\mathcal{R}$ on $\rho$}. The revenue is
proportional to $\rho/(1-\rho)$, therefore the revenue is increasing
in $\rho$. Heavier traffic will induce greater congestion, and hence,
jobs have to invest more in their purchase of priority in order to keep the
same performance. 

{\em Dependence of $\mathcal{R}$ on $\alpha$}. The revenue depends on $\alpha$ in three terms, 
and it seems that in general the effect of changing $\alpha$ in the last two terms 
is significantly smaller than
that of changing $\alpha$ in the first term $\rho/(\alpha(1-\rho))$.
Hence we would expect that the revenue is in general decreasing in $\alpha$.
The next result shows this intuition holds if $c_1/c_K$ is not too high.

\begin{theorem}\label{thm:rev_alpha}
The revenue $\mathcal{R}$ is decreasing in $\alpha>0$ if $c_1/c_K<e^4$.
\end{theorem}

On the other hand, $\mathcal{R}$ could be increasing in $\alpha$ in some cases.
For instance, if $K=2$ and $c_1/c_2$ is large enough, 
then $\partial\mathcal{R}/\partial\alpha$ is positive around $\alpha=1$ (see the proof of theorem for details).
To explain this special scenario where the monotonicity does not hold,
we first note that a smaller $\alpha$ in general induces a higher revenue
because jobs have incentive to purchase higher priority (as a response to the stronger diminishing marginal cost effect).
However, in the HTE,
significant asymmetry in costs will result in significant asymmetry in equilibrium priorities.
Therefore when $c_1/c_K$ is large, 
the optimal priorities already exhibit significant differences even
when $\alpha$ is not small, and thus 
in equilibrium at small $\alpha$, jobs have lower incentive to increase their priorities 
compared to what they do with mutually comparable costs.

With both (\ref{eq:efficiency}) and  (\ref{eq:revenue}), it is quite surprising to see that in the HTE,
\begin{equation*}
\text{total cost of all jobs} = \mathcal{C}=\alpha\mathcal{R}=\alpha\cdot\text{total revenue of the system}.
\end{equation*}
Thus, we obtain an interesting insight: {\em another interpretation of
  $\alpha$ is the users' equilibrium cost per unit revenue.}
We have shown that the user's total cost is increasing in $\alpha$,
(i.e., the system efficiency is decreasing in $\alpha$),
and the system revenue is decreasing in $\alpha$ under some mild conditions. Therefore, in a wide range of regimes,
from the standpoint of system manager, 
smaller $\alpha$ is more favorable in terms of both efficiency and revenue.
Note that smaller $\alpha$ is somewhat more ``unfair,'' however, as it approaches a strict priority system.

\section{Discussion and Conclusions}
\label{sec:dis}

We believe our work makes significant progress on two fronts.  First,
the DPS queueing model is important in its own right as a benchmark
model for analysis of priority pricing for shared resource services.
Our analysis provides extensive insight into this queueing system with strategic
behavior.  Second, and perhaps of greater longer term interest, our
approximation methodology suggests a broader research program for
understanding strategic  behavior in queueing systems: by exploiting
large system asymptotics, we can simplify {\em both} the complexity of
the stochastic system, as well as the complexity of the economic
system.

We conclude by discussing several extensions and open directions.

{\em Random order of service.}  Consider an alternative prioritized allocation policy, 
the {\em random order of service} (ROS) policy. 
In the ROS policy, only one job is served at a time and upon completion of this job, 
a new job starts to be served with probability proportional to its priority level.
Therefore, if there are currently
$N$ jobs waiting in the system and job $\ell$ has chosen priority level
$\beta_\ell$, then in the ROS policy, job $\ell$ is the next job to start service 
with probability $\beta_\ell/\sum_{m=1}^N\beta_m$.

Although ROS and DPS are different in the ways they allocate service among jobs,
the ratio of expected service allocated to two jobs is the same as
the ratio of their priorities in both schemes.
Therefore, we would expect some similarities in the expected processing times of jobs
in these two allocation policies.
In fact, Ayesta et al. \cite{aye11} show that in heavy traffic regime,
the expected processing time of a class $j$ job in a ROS system is exactly $V^{HT}(\beta_j; \boldsymbol{\beta})$, 
the same as in a DPS system.
Thus in the heavy traffic ROS system, 
the expected processing time of a job with priority $\beta$ is $V^{HT}(\beta; \bsy{\beta})$.
Therefore 
all our results on heavy traffic equilibria of the DPS system also
hold for the ROS system.

{\em Class level games.} Based on the heavy traffic processing time approximation results in
Theorem~\ref{thm:heavy}, one can also propose a similar heavy traffic
equilibrium (HTE) concept for class level games. However, although the
processing time approximation allows us to greatly simplify the
computation of best response strategies, we are not able to obtain a
closed form expression for the class level HTE.  
More specifically, we can similarly define heavy traffic processing time by analogously defining 
\[
\gamma_{-i}(\beta; \boldsymbol{\beta})=\frac{1}{\rho}\left[\sum_{j=1,j\neq i}^K\frac{\rho_j}{\beta_j}+\frac{\rho_i}{\beta}\right].
\]
We can obtain a system of {\em non-linear} equations to compute this class level heavy traffic equilibrium, 
but we are not able to get any closed form expressions for it, 
mainly due to two features of $\gamma_{-i}(\beta; \bsy{\beta})$ 
that are different from $\gamma(\boldsymbol{\beta})$ of a job level game. 
First, $\gamma_{-i}(\beta; \bsy{\beta})$ is subscripted by $i$, which implies that different classes face 
different environments in the system. 
This is because, unlike the infinitesimal single user in job level problem, 
each class as a player is not negligible in the system. 
Second, $\gamma_{-i}(\beta; \bsy{\beta})$ explicitly depends on $\beta$, the action of class $i$,
due to the {\em intra-class externality} behavior: 
in a class level game, a class chooses one
priority level for {\em all} its jobs simultaneously; 
while in a job level game, a job can choose any priority level regardless of other jobs belonging to its class.

Nevertheless, the intra-class externality effect will
become negligible in a regime where the number of classes is large;
this is a particularly useful regime for computing services where the
number of users grows large, and each user is viewed as a
distinct class. This observation motivates us to consider
a limiting model in which the number of classes approaches infinity
and any single class becomes infinitesimal. Thus, it can be connected
to the job level game model. We show that in the large system
limit of the class level game,
the heavy traffic equilibrium exists, is unique, can be computed in closed form,
and is the limit of the finite class equilibrium as the number of
class goes to infinity. Detailed discussions and proofs are in the Appendix.

{\em Networks}.  Our model considered a single resource; more
generally, we can extend some of our basic results to a network
setting.  One approach to considering models with more general network
structure is as follows.  Consider a setting with multiple resources,
where each resource runs its own market, and serves according to the
DPS policy.  Arriving jobs are characterized by a workload vector:
completion of service requires simultaneous effort from all resources
in this requirement vector, sufficient to complete the corresponding
workload at that resource generated by this job.  In this case, each
resource market operates independently of the others, but they are
coupled through the utility functions of the jobs.  In particular, we
assume job will be sensitive to the {\em maximum} completion time
across all resources it considers.  Since the maximum of a collection
of convex functions is still convex, if we consider a heavy traffic
equilibrium of the network game, the objective function of a
single job will remain convex in its own bid vector.  Using this
insight we can prove existence of HTE in a similar manner to our
earlier development.  In addition, we can leverage
the price of anarchy bounds of Theorem \ref{thm:hte_poa} to derive
bounds on inefficiency for HTE of the network setting; our approach
here is similar to \cite{Joh04}, who prove inefficiency bounds for
network resource allocation games by reduction to single resource games.  
Details of the network model are in the appendix.

{\em Endogenous arrival rates}.  Some previous work (e.g.,
\cite{mend90,kim03}) on priority pricing in queueing systems allows
for strategic choice of the arrival rate.  In our model, this might
mean jobs only enter if their total cost (cost of waiting plus
payment) does not exceed a reservation utility.  A significant
challenge here is that when arrival rates are endogenized, heavy
traffic cannot be exogenously guaranteed.  However, we believe
approximating the waiting time may still yield valuable insight into
this game.  Characterizing the quality of approximate equilibria in
this regime remains an open direction.

\bibliographystyle{plain}

\bibliography{qg_refs}

\begin{thebibliography}{10}

\bibitem{allgur10}
G.~Allon and I.~Gurvich.
\newblock Pricing and dimensioning competing large-scale service providers.
\newblock {\em Manufacturing \& Service Operations Management}, 12:449--469,
  2010.

\bibitem{altavraye06}
E.~Altman, K.~Avrachenkov, and U.~Ayesta.
\newblock A survey on discriminatory processor sharing.
\newblock {\em Queueing Systems}, 53(1-2):53--63, 2006.

\bibitem{aye11}
U.~Ayesta, A.~Izagirre, and I.M. Verloop.
\newblock Heavy traffic analysis of the discriminatory random-order-of-service
  discipline.
\newblock {\em Performance Evaluation Review}, 39(2):41--43, September 2011.

\bibitem{bil87}
P.~Billingsley.
\newblock {\em Weak Convergence of Measures: Applications in Probability}.
\newblock Society for Industrial Mathematics, Philadelphia, PA, 1987.

\bibitem{chen10}
Y.~Chen, C.~Maglaras, and G.~Vulcano.
\newblock Design of an aggregated marketplace under congestion effects:
  Asymptotic analysis and equilibrium characterization.
\newblock {\em Working Paper}, 2010.

\bibitem{coxsmi61}
D.R. Cox and W.L. Smith.
\newblock {\em Queues}.
\newblock Methuen and Wiley, London and New York, 1961.

\bibitem{faymitias80}
G.~Fayolle, I.~Mitrani, and R.~Iasnogorodski.
\newblock Sharing a processor among many job classes.
\newblock {\em Journal of the ACM}, 27(3):519--532, 1980.

\bibitem{gla83}
A.~Glazer and R.~Hassin.
\newblock ?/m/1: On the equilibrium distribution of customer arrivals.
\newblock {\em European Journal of Operational Research}, 13:146--150, 1983.

\bibitem{hashav03}
R.~Hassin and M.~Haviv.
\newblock {\em To queue or not to queue: Equilibrium behavior in queueing
  systems}.
\newblock Kluwer Academic Publishers, 2003.

\bibitem{havwal97}
M.~Haviv and J.~van~der Wal.
\newblock Equilibrium strategies for processor sharing and queues with relative
  priorities.
\newblock {\em Probability in the Engineering and Informational Sciences},
  11(4):403--412, 1997.

\bibitem{jai11}
R.~Jain, S.~Juneja, and N.~Shimkin.
\newblock The concert queueing game: to wait or to be late.
\newblock {\em Discrete Event Dynamic Systems}, 21:103--134, 2011.

\bibitem{Joh04}
R.~Johari and J.~N. Tsitsiklis.
\newblock Efficiency loss in a network resource allocation game.
\newblock {\em Mathematics of Operations Research}, 29(3):407--435, 2004.

\bibitem{kankelleewil09}
W.~Kang, F.~Kelly, N.~Lee, and R.~Williams.
\newblock State space collapse and diffusion approximation for a network
  operation under a fair bandwidth sharing policy.
\newblock {\em The Annals of Applied Probability}, 19(5):1719--1780, 2009.

\bibitem{kelmautan98}
F.P. Kelly, A.~Maulloo, and D.~Tan.
\newblock Rate control in communication networks: shadow prices, proportional
  fairness and stability.
\newblock {\em Journal of the Operational Research Society}, 49(3):237--252,
  1998.

\bibitem{kim03}
Y.~J. Kim and M.~V. Mannino.
\newblock Optimalincentive-compatiblepricing for m/g/1queues.
\newblock {\em Operations Research Letters}, 31:459--461, 2003.

\bibitem{kin62}
J.F.C. Kingman.
\newblock On queues in heavy traffic.
\newblock {\em Journal of the Royal Statistical Society. Series B
  (Methodological)}, 24(2):383--392, 1962.

\bibitem{kle67}
L.~Kleinrock.
\newblock Time-shared systems: A theoretical treatment.
\newblock {\em Journal of {ACM}}, 14(2):242--261, 1967.

\bibitem{kle75}
L.~Kleinrock.
\newblock {\em Queueing Systems, Volume 1: Theory}.
\newblock Wiley--Interscience, New York, 1975.

\bibitem{masrob00}
L.~Massouli\'{e} and J.~Roberts.
\newblock Bandwidth sharing and admission control for elastic traffic.
\newblock {\em Telecommunication Systems}, 15(1-2):185--201, 2000.

\bibitem{mend90}
H.~Mendelson and S.~Whang.
\newblock Optimal incentive-compatible priority pricing for the m/m/1 queue.
\newblock {\em Operations Research}, 38:870--883, 1990.

\bibitem{MoWalrand00}
J.~Mo and J.~Walrand.
\newblock Fair end-to-end window-based congestion control.
\newblock {\em IEEE/ACM Trans. Netw.}, 8:556--567, October 2000.

\bibitem{nai11}
J.~Nair, A.~Wierman, and B.~Zwart.
\newblock Exploiting network effects in the provisioning of large scale
  systems.
\newblock {\em Proceedings of 29th International Symposium on Computer
  Performance, Modeling, Measurements and Evaluation.}, 2011.

\bibitem{nao69}
P.~Naor.
\newblock The regulation of queue size by levying tolls.
\newblock {\em Econometrica}, 37(1):15--24, 1969.

\bibitem{regsen96}
K.~Rege and B.~Sengupta.
\newblock Queue-length distribution for the discriminatory processor-sharing
  queue.
\newblock {\em Operations Research}, 44(4):653--657, 1996.

\bibitem{rei84}
M.I. Reiman.
\newblock Open queueing networks in heavy traffic.
\newblock {\em Mathematics of Operations Research}, 9(3):pp. 441--458, 1984.

\bibitem{Rosen65}
J~B Rosen.
\newblock Existence and uniqueness of equilibrium points for concave $n$-person
  games.
\newblock {\em Econometrica}, 33(3):520--534, July 1965.

\bibitem{verayenun10}
I.M. Verloop, U.~Ayesta, and R.~Nunez-Queija.
\newblock Heavy-traffic analysis of a multiple-phase network with
  discriminatory processor sharing.
\newblock {\em Operations Research}, 59(3):648--660, 2011.

\end{thebibliography}


\appendix



\section{Proofs of Theorems}

\subsection{Proof of Theorem~\ref{thm:vmulti}}

We follow the same approach as in \cite{havwal97} (see also Theorem~4.8 in \cite{hashav03}) by decomposing
the problem into two parts. Fix a job with priority
$\beta > 0$, and fix $\boldsymbol{\beta}=(\beta_1, \cdots, \beta_K) > 0$
where $\beta_i$ is the priority level of class $i$ jobs. Let
$U(\beta; \boldsymbol{\beta}, n_1, \cdots, n_K)$ be the expected time in the
system if this job has priority $\beta$, and is currently being processed
with $n_i$ class $i$ jobs in the system. Then
conditional on the first transition, we have the following recursion:
\begin{equation}\label{eq:recursion}
\begin{aligned}
U(\beta; \boldsymbol{\beta}, n_1, \cdots,n_K)&=
\frac{1}{\lambda+\mu}+\sum_{i=1}^K\frac{\lambda_i}{\lambda+\mu}U(\beta;
\boldsymbol{\beta}, n_1, \cdots, n_i+1, \cdots, n_K)\\
&+\sum_{i=1}^K\frac{\mu}{\lambda+\mu}\frac{n_i\beta_i}{\sum_{k=1}^K n_k\beta_k+\beta} U(\beta; \boldsymbol{\beta}, n_1,\cdots, n_i-1, \cdots, n_K).
\end{aligned}
\end{equation}
Following the same argument as in \cite{hashav03}, we can show that
$U(\beta, n_1, \cdots, n_K)$ is linear in each $n_i$. That is, there
are functions $U_i(\beta; \boldsymbol{\beta})~(i=0, \cdots, K)$ such that:
\begin{equation}\label{eq:linear}
U(\beta; \boldsymbol{\beta}, n_1, \cdots, n_K)=U_0(\beta; \boldsymbol{\beta})+\sum_{i=1}^KU_i(\beta; \boldsymbol{\beta})n_i.
\end{equation}
Substituting (\ref{eq:linear}) into (\ref{eq:recursion}) 
and comparing the coefficients of each $n_i$ as well as the constant term, we obtain:
\beqa
\beta_k\left[1+\sum_{i=1}^K\lambda_iU_i(\beta; \boldsymbol{\beta})\right]
&=&\mu(\beta+\beta_k)U_k(\beta; \boldsymbol{\beta}), \quad k=1, \cdots, K;\\
1+\sum_{i=1}^K\lambda_iU_i(\beta; \boldsymbol{\beta}) &=&\mu U_0(\beta; \boldsymbol{\beta}).
\eeqa
Hence the solution to
(\ref{eq:recursion}) is
\[U_i(\beta; \boldsymbol{\beta})=\frac{\beta_i}{\beta_i+\beta}U_0(\beta; \boldsymbol{\beta}),\ i=1, \cdots, K;\ \ 
 U_0(\beta; \boldsymbol{\beta})=\left[\mu-\sum_{i=1}^K\frac{\lambda_i\beta_i}{\beta_i+\beta}\right]^{-1}.
\]
The above equation, along with (\ref{eq:linear}) gives the expression
for $V(\beta; \boldsymbol{\beta})$ in terms of the expected queue
lengths.

\subsection{Proof of Theorem~\ref{thm:ne_existence}}

It is straightforward to show that $V(\beta;\bsy{\beta}) \leq
1/(\mu(1-\rho))$: the expected processing time of any job, regardless of priority,
cannot be longer than the expected length of a busy period in an M/M/1
queue with arrival rate $\lambda$ and service rate $\mu$ \cite{kle75}.  (This
follows because the discriminatory processor sharing policy is
work-conserving, so the length of a busy period will be identical to
that in an M/M/1 queue.)  It then follows that there exists an upper bound
$\overline{\beta}$ such that no job ever has $\beta >
\overline{\beta}$ as an optimal strategy.  Therefore, we can restrict
the strategy space of every job to the compact set
$[\underline{\beta}, \overline{\beta}]$.

 
Next we show that $V(\beta; \bsy{\beta})$ is convex in $\beta$. 
Combining 
(\ref{eq:vmulti}) and (\ref{eq:uk}) 
gives $V(\beta; \bsy{\beta})=f(\beta)/g(\beta)$, where 
$f(\beta)=1+\sum_{i=1}^K\frac{\lambda_iW_i\beta_i}{\beta_i+\beta}$ and $g(\beta)=\mu-\sum_{i=1}^K\frac{\lambda_i\beta_i}{\beta_i+\beta}$.
It is easy to check that $f(\beta)>0, f'(\beta)<0, f''(\beta)>0$; and
$g(\beta)>0, g'(\beta)>0, g''(\beta)<0$, 
thus
\[
\left[\frac{f(\beta)}{g(\beta)}\right]''=\frac{
\left[f''(\beta)g(\beta)-f(\beta)g''(\beta)\right]}{[g(\beta)]^2}
-\frac{2\left[f'(\beta)g(\beta)-f(\beta)g'(\beta)\right]g(\beta)g'(\beta)
}{[g(\beta)]^4}>0.
\]
Moreover, $\beta^{\alpha}$ is convex in $\beta$ since $\alpha\geq 1$, 
so $c_iV(\beta; \bsy{\beta})+\beta^{\alpha}$ is convex in $\beta$. 
It follows by Rosen's existence theorem \cite{Rosen65} that a pure Nash
equilibrium exists for the game in this case.

\subsection{Proof of Theorem~\ref{thm:ht_equi}}


We note that the best response of a class $i$ job given that all class $j$ jobs choose $\beta_j$ $(j = 1, 2, \ldots, K)$ is
\[
\beta_i(\bsy{\beta})=\arg\min_{\beta \geq 0} \left(c_i V^{HT}(\beta, \bsy{\beta}) + \beta^\alpha \right) = \arg\min_{\beta \geq 0} \left(\frac{c_i}{\mu(1-\rho)\beta\gamma(\bsy{\beta})}+ \beta^\alpha \right).
\]
The first order condition of optimality yields a unique solution:
\beqn
\beta_i^*(\bsy{\beta})=\left(c_i^{-1}\alpha\mu(1-\rho)\gamma(\bsy{\beta})\right)^{-\frac{1}{\alpha+1}}.
\label{eqn:first}
\eeqn
And the second derivative of the objective function at this point is
\[
\frac{2c_i}{\mu(1-\rho)\gamma(\bsy{\beta})}\frac{1}{(\beta^*_i)^3}+\alpha(\alpha-1)(\beta_i^*)^{\alpha-2}
=\frac{(\alpha+1)c_i}{\mu(1-\rho)\gamma(\bsy{\beta})}\frac{1}{(\beta_i^*)^3}> 0.
\]
Therefore, (\ref{eqn:first}) is the unique minimizer of the objective function.
Recall that 
$\gamma(\bsy{\beta}) = \sum_{i=1}^K\rho_i/(\rho \beta_i)$.
Thus, at the equilibrium,
\begin{equation}
\gamma(\bsy{\beta}^*)=\sum_{i=1}^K\rho_i\rho^{-1}\left(c_i^{-1}\alpha\mu(1-\rho)\gamma(\bsy{\beta}^*)\right)^{\frac{1}{\alpha+1}}
\Rightarrow \gamma(\bsy{\beta}^*)=(S_1/\lambda)^{\frac{\alpha+1}{\alpha}}\left(\alpha\mu(1-\rho)\right)^{\frac{1}{\alpha}},
\label{eqn:second}
\end{equation}
where $S_1 = \sum_{i=1}^K \lambda_ic_i^{-\frac{1}{\alpha+1}}$. Plugging (\ref{eqn:second}) into (\ref{eqn:first}) yields the result:
\[
\beta_i^*=c_i^{\frac{1}{\alpha+1}} [\lambda^{-1}\alpha \mu (1-\rho)S_1]^{-\frac{1}{\alpha}} 
= c_i^{\frac{1}{\alpha+1}} [\rho^{-1}\alpha(1-\rho)S_1]^{-\frac{1}{\alpha}}.
\]
Therefore, the heavy-traffic equilibrium always exists, is unique, and can be calculated by the above closed form expressions.

\subsection{Proof of Theorem~\ref{thm:epsilon_approx}}

First, we observe that for any $\beta > 0$ and $\bsy{\beta} = (\beta_1, \cdots, \beta_K) > 0$, $V(\beta;\bsy{\beta})$ and $V^{HT}(\beta; \bsy{\beta})$ depend on $\beta$ and $\bsy{\beta}$ only through the ratios $\beta_i/\beta$ and $\beta_i/\beta_j$ for any $i,j$.

\noindent Now, since $\boldsymbol{\beta}^{(n)}$ is the HTE at the $n$-th system, for any $\delta \geq 0$, we have that
\begin{equation}\label{eq:asymhte}
c_iV_{(n)}^{HT}(\beta_i^{(n)}; \boldsymbol{\beta}^{(n)})
+(\beta_i^{(n)})^{\alpha}
\leq c_iV_{(n)}^{HT}\left(\delta\beta_i^{(n)}; \boldsymbol{\beta}^{(n)}\right)
+\left(\delta\beta_i^{(n)}\right)^{\alpha}.
\end{equation}
Define $\theta_j=\beta_j^{(n)}/\beta_i^{(n)}$ ($j=1, \cdots, K$). Then (\ref{eq:opt_beta_tilde}) implies that $\theta_j=\beta_j^{(n)}/\beta_i^{(n)}$ is independent of $n$ and the load $\rho^{(n)}$.
Therefore,
$V_{(n)}^{HT}(\beta_i^{(n)}; \boldsymbol{\beta}^{(n)})=V^{HT}_{(n)}(1; \boldsymbol{\theta})$ 
and 
$V_{(n)}(\beta_i^{(n)}; \boldsymbol{\beta}^{(n)})=V_{(n)}(1; \boldsymbol{\theta})$. 
And thus,
\begin{equation}\label{eq:asym1}
\begin{aligned}
&\lim_{n\rightarrow\infty} (1-\rho^{(n)}) [V_{(n)}^{HT}(\beta_i^{(n)}; \boldsymbol{\beta}^{(n)})-
V_{(n)}(\beta_i^{(n)}; \boldsymbol{\beta}^{(n)})]\\
=&\lim_{n\rightarrow\infty} (1-\rho^{(n)}) [V^{HT}_{(n)}(1; \boldsymbol{\theta})-
V_{(n)}(1; \boldsymbol{\theta})]=0.
\end{aligned}
\end{equation}
The last equality follows from the asymptotic exactness of $V^{HT}$.
Similarly, 
$\beta_j^{(n)}/(\delta\beta_i^{(n)})$ is also independent of $n$ and system load,
so we have that
\begin{equation}\label{eq:asym2}
\lim_{n\rightarrow\infty} (1-\rho^{(n)}) 
\left[V_{(n)}^{HT}\left(\delta\beta_i^{(n)}; \boldsymbol{\beta}^{(n)}\right)
- V_{(n)}\left(\delta\beta_i^{(n)}; \boldsymbol{\beta}^{(n)}\right)\right]=0.
\end{equation}
Then the claim follows by plugging (\ref{eq:asym1}) and $(\ref{eq:asym2})$ into (\ref{eq:asymhte}).

\subsection{Proof of Proposition~\ref{thm:eff_alpha}}

Take the first derivative of $\mathcal{C}$ with respect to $\alpha$, we have that
\begin{equation*}
\frac{\partial\mathcal{C}}{\partial\alpha}=\frac{\rho}{(1-\rho)(\alpha+1)^2}\sum_{i<j}
\left[\lambda_i\lambda_j(c_i^{\frac{\alpha}{\alpha+1}}c_j^{-\frac{1}{\alpha+1}}
-c_j^{\frac{\alpha}{\alpha+1}}c_i^{-\frac{1}{\alpha+1}})\ln\frac{c_i}{c_j}\right]\left(\sum_i\lambda_ic_i^{-\frac{1}{\alpha+1}}\right)^{-2}.
\end{equation*}
Since $c_i^{\frac{\alpha}{\alpha+1}}c_j^{-\frac{1}{\alpha+1}}
-c_j^{\frac{\alpha}{\alpha+1}}c_i^{-\frac{1}{\alpha+1}}$ and $\ln\frac{c_i}{c_j}$ have the same signs,
this derivative is positive and $\mathcal{C}$ is increasing in $\alpha$.

\subsection{Proof of Theorem \ref{thm:hte_poa}}

To compute the system processing cost of the $c\mhyphen\mu$ rule,
consider the following M/M/1 queue models:
Class 1 jobs have the highest priority in the system 
and themselves form an M/M/1 queue with parameter $(\lambda_1, \mu)$, 
therefore basic M/M/1 queue result \cite{kle75} implies the expected number of class 1 jobs in the system is
$\bE[N_1]=\rho_1/(1-\rho_1)$.
Then Little's law implies the expected processing time for class 1 job is $\bE[N_1]/\lambda_1$.
If we further consider both class 1 and class 2 jobs, they are not preempted by any other jobs in the system,
therefore they form yet another M/M/1 queue with parameter $(\lambda_1+\lambda_2, \mu)$,
and thus the expected number of class 2 jobs in the system is
$\bE[N_2]=(\rho_1+\rho_2)/(1-\rho_1-\rho_2)-\bE[N_1]$.
In general, the expected number of class $i$ jobs in the system is
$\bE[N_i]=\frac{\sum_{j=1}^i\rho_j}{1-\sum_{j=1}^i\rho_j}-\frac{\sum_{j=1}^{i-1}\rho_j}{1-\sum_{j=1}^{i-1}\rho_j}$. Finally, the system processing cost of the system with $c\mhyphen\mu$ rule is 
\begin{equation}\label{eq:opt_efficiency}
\mathcal{C}^{opt}=\sum_{i=1}^Kc_i\lambda_i(\bE[N_i]/\lambda_i)=\sum_{i=1}^Kc_i\left(\frac{\sum_{j=1}^i\rho_j}{1-\sum_{j=1}^i\rho_j}-\frac{\sum_{j=1}^{i-1}\rho_j}{1-\sum_{j=1}^{i-1}\rho_j}\right).
\end{equation}

To bound the PoA, we first note that $c_i$ is decreasing in $i$
and the expected number of all jobs in the system is $\rho/(1-\rho)$, therefore
$\mathcal{C}^{opt} > c_K \sum_{i=1}^K \bE[N_i] = c_K\rho/(1-\rho)$.
On the other hand, let $c'_i=c_i/c_K$ and $\lambda'_i=\lambda_i/\lambda_K$ for $i=1, \cdots, K-1$,
then 
\begin{equation}\label{eq:opt_eff_bound_2}
\frac{\mathcal{C}}{\mathcal{C}^{opt}}<\frac{\mathcal{C}(1-\rho)}{c_K\rho}=
\frac{  \sum_{i=1}^{K-1} \lambda'_i{c'}_i^{\frac{\alpha}{\alpha+1}}  +1  }
{  \sum_{i=1}^{K-1}  \lambda'_i{c'}_i^{-\frac{1}{\alpha+1}}  +1  }
<  \left(\sum_{i=1}^{K-1}\lambda'_i \right)  {c'}_1^{\frac{\alpha}{\alpha+1}}+1.
\end{equation}
Therefore $\dfrac{\mathcal{C}}{\mathcal{C}^{opt}}<\dfrac{\lambda-\lambda_K}{\lambda_K}\left(\dfrac{c_1}{c_K}\right)^{\frac{\alpha}{\alpha+1}}+1$.

We note that there exist some systems in which the PoA of HTE can be made arbitrarily large. Here is an example. Let $\lambda_i$'s all equal, set $c_i=mc_K$ for $i=1, \cdots, K-1$ and $\rho=1-m^{-2}$. Then $\sum_{j=1}^i\rho_j<1-\frac{1}{K}$ for $i<K$ and $\rho/(1-\rho)=m^2-1$. We have that
\[
\mathcal{C}^{opt}<c_K(m^2+mK), \ \ \mathcal{C}=\frac{c_K\rho}{(1-\rho)}\frac{(K-1)m^{\frac{\alpha}{\alpha+1}}+1}{(K-1)m^{-\frac{1}{\alpha+1}}+1}.
\]
Hence,
\[
\frac{\mathcal{C}}{\mathcal{C}^{opt}}>\frac{(K-1)m^{\frac{\alpha}{\alpha+1}}+1}{(K-1)m^{-\frac{1}{\alpha+1}}+1}\frac{m^2-1}{m^2+mK}.
\]
Letting $m$ go to infinity makes the PoA arbitrarily large. Also, note that in this case the PoA bound is $(K-1)m^{\frac{\alpha}{\alpha+1}}+1$, which means that the PoA bound is ``asymptotically tight''.

\subsection{Proof of Theorem~\ref{thm:rev_alpha}}

Take the first derivative of $\mathcal{R}$ with respect to $\alpha$, we have that
\begin{align}\label{eq:R_derivative}
\frac{\partial\mathcal{R}}{\partial\alpha}&=\frac{\rho}{(1-\rho)\alpha^2}\left(\sum\lambda_ic_i^{-\frac{1}{\alpha+1}}\right)^{-2}\nonumber\\
&\times\left[\frac{\alpha}{(\alpha+1)^2}\sum_{i<j}
\lambda_i\lambda_j(c_i^{\frac{\alpha}{\alpha+1}}c_j^{-\frac{1}{\alpha+1}}
-c_j^{\frac{\alpha}{\alpha+1}}c_i^{-\frac{1}{\alpha+1}})\ln\frac{c_i}{c_j}
-\sum_{i,j}\lambda_i\lambda_jc_i^{\frac{\alpha}{\alpha+1}}c_j^{-\frac{1}{\alpha+1}}\right].
\end{align}
If $c_1/c_K<e^4$, then $\max_{i,j}\ln(c_i/c_j)\leq 4\leq \frac{(\alpha+1)^2}{\alpha}$.
It follows from (\ref{eq:R_derivative}) that $\partial\mathcal{R}/\partial\alpha$ is negative
and $\mathcal{R}$ is decreasing in $\alpha$.

\subsection{Class Level Game and Approximation}

As what we did in 
Section \ref{sec:heavy},
we can define $W^{HT}_i(\bsy{\beta})=1/\mu(1-\rho)\beta_i\gamma(\bsy{\beta})$.
It then follows from 
Theorem~\ref{thm:heavy}
and Little's law that
\[
\lim_{\rho\rightarrow1}(1-\rho)W_i(\bsy{\beta})=\lim_{\rho\rightarrow1}\frac{\bE[N_i]}{\lambda_i}=\frac{1}{\lambda\beta_i\gamma(\bsy{\beta})}.
\]
Therefore $W^{HT}$ is asymptotically exact: as $\rho\rightarrow1$,
\[
(1-\rho)[W^{HT}_i(\bsy{\beta})-W_i(\bsy{\beta})] \rightarrow 0.
\]
Similarly, we can define $\boldsymbol{\beta} =
(\beta_1, \cdots, \beta_K)$ as a class level heavy traffic equilibrium if $\forall i = 1, \cdots, K,$
\begin{align*}
\beta_i&=\arg\min_{\beta>0}~c_iW_i^{HT}(\beta_1, \cdots , \beta_{i-1}, \beta, \beta_{i+1}, \cdots, \beta_K)+\beta^{\alpha}\\
&=\arg\min_{\beta>0}~\frac{c_i \rho}{\mu(1-\rho)}\frac{1}{\beta}\left[\sum_{j\neq i}\frac{\rho_j}{\beta_j}+\frac{\rho_i}{\beta}\right]^{-1}+\beta^{\alpha}.
\end{align*}
We can obtain a system of {\em non-linear} equations to compute this class level heavy traffic equilibrium, but we are not able to get any closed form expressions for it, mainly due to the existence of intra-class externality: $\gamma(\bsy{\beta})$ depends on 
all priorities of all classes, including the class who is optimizing.
This is fundamentally different from the job level game, 
where for any individual job, 
its own priority is infinitesimal and does not affect the auxiliary variable $\gamma(\bsy{\beta})$.

To eliminate the externality within each class, 
we consider a limiting model where the number of
classes approaches infinity. 
Then each class as a whole is infinitesimal in the system,
and hence, the system stays almost the same even if we exclude a whole
class.
Next we formalize this idea. We consider a series of systems indexed by
$K$ satisfying the following properties.

\begin{definition}
A {\bf limiting class level game} consists of the following elements.
\begin{enumerate}
\item There are $K$ classes in system $K$, each with cost $c_k^K$
arrival rate $\lambda_k^K$, $k=1, \cdots, K$, and service
requirement rate $\mu^K$.
\item The costs $c_k^K$ are i.i.d.~samples from a distribution
  $F(\cdot)$ with a positive and closed support $S_c$.
\item The arrival rates $\lambda_k^K$ are i.i.d~samples from a
  distribution $G(\cdot)$ with a positive and closed support
  $S_{\lambda}$ (independent of the costs).  Let $\bE_G[\lambda] =
  \int_{S_\lambda} \lambda dG$, the limiting mean arrival rate.
\item The service rate $\mu^K$ in system $K$ is chosen to ensure that
  the system is stable, i.e., $\rho^K \triangleq \sum_{k = 1}^K
  \lambda_k^K/\mu < 1$.
\item $\lim_{K \to \infty} \mu^K/K$ exists; denoting this limit by
  $\mu$, we assume that $\bE_G[\lambda] < \mu$.
\item Feasible priority levels are bounded by a positive and closed
support $S_{\beta}=[\overline{\beta}, \underline{\beta}]$.
\end{enumerate}
\end{definition}
As $K\rightarrow\infty$, the limit of these systems consists of a
continuum of classes with independent cost and arrival rate
distributions as $F$ and $G$, and with limiting per class service rate
$\mu$.

Suppose we are given a limiting class level game.  Let $B:S_c\rightarrow
\reals^+$ denote a {\em priority strategy function}, that maps the unit cost of
a class (and its jobs) to the priority level chosen by this class
(and its jobs).  Let $V(\beta; B, F, G, \mu)$ be the expected
steady state processing time of a job with priority $\beta$, processed
with a continuum of classes of jobs characterized by priority
strategy function $B$, cost distribution $F$, arrival rate
distribution $G$, and per class service rate $\mu$. Then the symmetric
Nash equilibrium for this continuum game is a strategy function
$B(\cdot)$ such that:
\[
B(c)=\arg\min_{\beta\geq 0}~ c V(\beta; B , F,
G)+\beta^{\alpha},\ \ \forall\ c\in S_c.
\]
Given the complexity of $V$ in general, this cannot be solved in
closed form.

\subsubsection{Approximate Processing Time}

Let us define 
\begin{equation}\label{eq:cl:betaht}
\gamma(B)=\int_{S_c}\int_{S_\lambda}\lambda/(\mu\rho B(c))dFdG=\int_{S_c}B(c)^{-1}dFdG,
\end{equation}
where the last equality comes from the independence between $\lambda$ and $c$, and define
\begin{equation}\label{eq:cl:vht}
V^{HT}(\beta; B)=\frac{1}{\mu(1-\rho)}\frac{1}{\gamma(B)\beta}.
\end{equation}
Then we have
similar approximation result in heavy traffic for the limiting class level game, which shows that $V^{HT}$ is asymptotically exact.

\begin{proposition}\label{thm:approx2}
Suppose we are given a limiting
class level game, and any positive strategy function $B : S_c\rightarrow
S_{\beta}$.  Then for all $\beta > 0$:
\[
\lim_{\rho \to 1} (1 - \rho) V^{HT}(\beta, B)=\lim_{K
  \to \infty} \lim_{\rho^K \to 1} (1-\rho^K)W_1(\beta, B(c_2), \cdots,
B(c_K)).
\]
\end{proposition}

\begin{proof}


By definition (\ref{eq:cl:betaht})  and (\ref{eq:cl:vht})we have
\begin{equation}\label{eq:limit2:v2}
\lim_{\rho\to 1}(1-\rho)\lambda V^{HT}(\beta;
B)=\left[\beta\int_{S_c}B(c)^{-1}dF\right]^{-1}.
\end{equation}

On the other hand, it follows from 
Theorem~\ref{thm:heavy} 
and
Little's law that
\[
\lim_{\rho^K\rightarrow1}(1-\rho^K)\lambda^K W_1(\beta,
B(c_2), \cdots, B(c_K))=\frac{1}{\gamma^K(\beta)\beta}.
\]
where
$\ds \gamma^K(\beta)=\frac{\rho_1}{\rho\beta}+\sum_{i=2}^K\frac{\rho_i}{\rho B(c_i)}$.
The strong law of large numbers implies
\beqn
\lim_{K\rightarrow\infty}\gamma^K(\beta)
=\lim_{K\rightarrow\infty}\left(\frac{\lambda_1}{\lambda^K\beta}+\sum_{i=2}^K\frac{\lambda_i}{\lambda^KB(c_i)}\right)
=\int_{S_c}B(c)^{-1}dF.
\eeqn
Therefore,
\begin{equation}\label{eq:limit2:v3}
\ds \lim_{K\rightarrow\infty}\lim_{\rho^K\rightarrow1}(1-\rho^K)\lambda^K W_1(\beta,B(c_2), \cdots,B(c_K))
=\left[\beta\int_{S_c}B(c)^{-1}dF\right]^{-1}.
\end{equation}
Equations (\ref{eq:limit2:v2}) and (\ref{eq:limit2:v3}) together complete the proof.

\end{proof}



\subsubsection{Class Level Equilibrium}

Inspired by heavy traffic equilibrium for the job level game,
we can also define the class level heavy traffic equilibrium
for the limiting class level game.

\begin{definition}
\label{def:cl_AE}
Suppose we are given a limiting class level game. A
{\bf class level heavy traffic equilibrium} is characterized by a positive
priority strategy function $B$
\begin{equation}
B(c)=\arg\min_{\beta\geq0}\left(c V^{HT}(\beta, B)+\beta^{\alpha}\right), \quad \forall c\in S_c,
\end{equation}
where $V^{HT}(\beta,B)$ is defined as in (\ref{eq:cl:vht}), with $\rho
= \bE_G[\lambda]/\mu$.
\end{definition}

Comparing this definition to that in the finite case, we
note that the summation over all classes is replaced by integration
over the support space,
and individual strategies are replaced by a strategy function,
since now we have a continuum of classes.

This class level heavy traffic equilibrium can also be solved in closed form.
\begin{proposition}
\label{thm:cl_eq}
The class level heavy traffic equilibrium for a limiting class level game  always exists and is unique. Moreover, it can be calculated in closed form as follows:
\begin{equation}\label{eq:cl:opt_beta}
B(c)=c^{\frac{1}{\alpha+1}}(\mu\alpha(1-\rho)S_2)^{-\frac{1}{\alpha}},
\end{equation}
where $S_2=\int_{S_c}c^{-\frac{1}{\alpha+1}}dF$ is independent of $\rho$.
\end{proposition}


The proof is just a similar repetition of the proof for 
Theorem~\ref{thm:ht_equi} 
and is omitted. 

Moreover, we can relate the HTE of the class
level game to the HTE of the job level game in a
series of finite systems: we show that a class level
game with infinitely many classes behaves like a job level game,
validating that the in-class externality has been mitigated.
If $K$ is fixed, system $K$ is a finite class system. The optimal
strategies in heavy traffic equilibrium are given by 
(\ref{eq:opt_beta_tilde}).
Now send $K\rightarrow\infty$, we find that the series of heavy 
traffic equilibrium strategies converges almost surely and the 
limit coincides with the corresponding strategy given by  the class level
heavy traffic equilibrium  function $B^*$, which can be easily verified 
by applying the strong law of large numbers in 
(\ref{eq:opt_beta_tilde}).

\begin{proposition}\label{thm:large_number}
Suppose we are given a limiting class level game.
Let $\beta_i^{K}$ denote the strategy of a class $i$ job used in heavy traffic equilibrium in the $K$th system,
and $B(\cdot)$ be the equilibrium strategy function used in the limiting class level game, then
\[
\lim_{K \to \infty} \left(\beta^{K}_i-B(c_i^K)\right)=0.
\]
\end{proposition}

\subsection{Network Model}
In this section, we provide a description of a network generalization
of our model.  Consider a setting with $J$ resources, and $K$
classes.  Jobs of class $i$ arrive at rate $\lambda_i$.  Each class
requires service from a subset of the resources; 
in particular, let $r_i$ denote the subset of resources that are used
by a job of class $i$.  Each job generates an exponentially
distributed workload with mean $1/\mu_{j}$ at resource $j$; let
$\rho_{ij} = \lambda_i/\mu_j$ be the traffic intensity of class $i$
at resource $j$.

We assume that each resource operates as an {\em independent market}.
In other words, each job bids independently at each resource, and each
resource uses the DPS policy to allocate resources to jobs.  Let
$\beta_{ij}$ be the bid of a class $i$ job at resource $j$; for
simplicity in this section, we assume that the payment of the job to
that resource is equal to $\beta_{ij}$ (i.e., that $\alpha = 1$ at
each resource.

Finally, in the model we consider, we assume that each job {\em
  simultaneously} requires service from each of the resources it
demands.  We assume that each resource is sensitive to its {\em
  maximum} waiting time across the resources.  This might be a
reasonable model if, for example, the resources correspond to
resources used to farm out parallelized jobs; in that case, the user
would be sensitive to the completion time of the slowest job run on
the resources they demand.

We have the following equivalent definition of a Nash equilibrium.

\begin{definition}\label{def:jobeqnet}
A {\bf Nash equilibrium} of the network game
consists of a class priority vector $\boldsymbol{\beta}=(\beta_{ij}, i =
1,\ldots, K; j \in r_i)$ such that
\begin{equation}\label{eq:def:jobeqnet}
\v{\beta}_i =\arg\min_{\v{\beta}>0}\left[c_i \max_{j \in r_i} \{ V_{j}(\beta_{ij};
\boldsymbol{\beta}^{(j)})\} +\sum_{j \in r_i}\beta_{ij}\right],\ \forall\ i=1, \cdots, K,
\end{equation}
where $\v{\beta}^{(j)} = (\beta_{kj}, k$ such that $j \in r_k)$ is the
class priority vector of jobs that use resource $j$.
\end{definition}

Here $V_{j}(\beta ; \v{\beta}^{(j)})$ is the waiting time in a DPS system for a
job with priority $\beta$ at
resource $j$, when the class priority vector at resource $j$ is
$\v{\beta}^{(j)}$, as before.  We can analogously define a {\em heavy 
  traffic equilibrium} (HTE) of the network game, by replacing $V_{j}$
by $V_{j}^{HT}$.  Note that this notion is formally justified if we consider a heavy
traffic limit where
$\rho_{ij}$ converges to a limit $\bar{\rho}_{ij}$, such that $\sum_{i
  : j \in i} \bar{\rho}_{ij} = 1$ for all $j$.  (In particular this
corresponds to the limit where {\em every} resource approaches
to heavy traffic simultaneously.)

If we let $\bar{V}_{j}$ denote the waiting time of a particular class
$i$ job at
resource $j$, then observe that $\max_{j \in 
r_i} \{ \bar{V}_{j} \}$ is a convex function of 
$(\bar{V}_{j}, j \in r_i)$.  As a result, the objective function of user
$i$ in the definition of heavy traffic equilibrium can be shown to be
convex, and so by standard arguments it is straightforward to show
that a HTE exists for this game; for brevity we omit the details.
Unfortunately, due to the complexity of the network setting, it is not
possible in general to establish either uniqueness of the equilibrium
or compute the equilibrium in closed form.

However, we can use our earlier results to obtain a bound on the price
of anarchy in this model.  We require one additional piece of
notation.  Let $b_{j}(\bar{V}_{j}; \v{\beta}^{(j)})$ be the value of
$\beta$ that ensures that the heavy traffic waiting time of a job at
resource $j$ is $\bar{V}_{j}$, when the class priority vector of other jobs
at resource $j$ is $\v{\beta}^{(j)}$.  In other words, let $b_{j}(\bar{V}_{j},
\v{\beta}^{(j)})$ be the
solution $\beta$ to:
\[ \bar{V}_{j} = V_j^{HT}(\beta ; \v{\beta}^{(j)}). \]
Though $b_{j}(\cdot)$ can be computed in closed form, the solution is
tedious and not particularly insightful.  For our purposes, all we
require is that it is convex, decreasing, and differentiable in
$\bar{V}_{j} > 0$.

We can then prove the following theorem.
\begin{theorem}
\label{th:networkpoa}
Suppose $\lambda_i = \lambda$ for all $i$.  Let $\v{\beta}^*$ be an HTE
of the network game, and let $V_{ij}^* = V_j(\beta_{ij}^*,
\v{\beta}^{(j)*})$.  Further, define:
\[ b_{ij}' = \frac{\partial b_j}{\partial \bar{V}_j}( V_{ij}^*,
\v{\beta}^{(j)*}).\]
Then the price of anarchy, i.e., the ratio of the 
HTE processing cost to the minimal system processing cost is bounded above by:
\[ (K-1) \max_j \sqrt{ \frac{\max_{i : j \in r_i}
    (-b_{ij}')}{\min_{i : j \in r_i} (-b_{ij}') }}+1. \]
\end{theorem}

Due to space constraints, we only provide a proof sketch here.  Our
proof technique follows the analysis of the price of anarchy of the
network game in \cite{Joh04}.  In that paper, it is shown
that by a decomposition approach, the price of anarchy of a network
resource allocation game can be studied by reduction to the price of
anarchy of a collection of single resource games.  In particular, we
write the processing time cost function of a user as a function of her waiting times
as:
\[ C_i(\v{V}_i) = c_i \max_{j \in r_i}  V_{ij}. \]
Now consider a new game where the cost function of user $i$ is instead given by:
\[ \hat{C}_i(\v{V}_i) = \sum_{j \in r_i} (-b_{ij}') ( V_{ij} - V_{ij}^*) + c_i \max_{j\in r_i}
 V_{ij}^*. \]
This cost function is derived by linearizing around the waiting time
vector observed in equilibrium.  It has two important properties:
first, $C_i(\v{V}_i^*) = \hat{C}_i(\v{V}_i^*)$; and second, because
of convexity of the original cost function, the first order condition
for optimality at the equilibrium can be used to show that for any
$\v{V}_i$, there holds $\hat{C}_i(\v{V}_i) \leq C_i(\v{V}_i)$.  In
particular, the optimal system processing cost can only be lower under
the new cost function. 

Next, observe that in equilibrium, since a user minimizes
$C_i(\v{V}_i)+\sum_{j\in r_i}b_j(V_{ij}, \v{\beta}^{(j)})$, 
the directional derivative of
$C_i(\v{V}_i^*)$ in the direction of the vector $(1, \ldots, 1)$ must
be equal to $c_i$, since the job must have equalized its waiting times at
the different resources in equilibrium.   The first order condition
can then be used to conclude that $\sum_{j \in r_i} -b_{ij}' =
c_i$.  (Note that $-b_{ij}' \geq 0$ since $b_j$ is decreasing
in waiting time.)  This ensures that $\hat{C}_i(\v{V}_i) > 0$ for all
$i$ and feasible $\v{V}_i$, and in particular that $c_i \max_j 
V_{ij}^* - \sum_j (-b_{ij}')(V_{ij}^*) \geq 0$.  

Finally, we note that $\hat{C}$ is {\em linear} in the waiting times
at different resources; thus the first order conditions for optimality
for job $i$ decompose across the resources.  Thus if we consider now
{\em independent} games at each resource $j$, where player $i$ with $j
\in r_i$ plays
with unit time cost  $-b_{ij}'$, then it follows that a
HTE for that game would also be $\boldsymbol{\beta}^{(j)*}$.  
 Since the equilibrium actions are the same in the network game and in
 the independent single server games, 
while the optimal social cost is lower in the latter,
the result then
follows using the same argument as in \cite{Joh04} by using
the price of anarchy bound for single resource games established in
Theorem 8.






\end{document}